\documentclass[twocolumn]{autart}

\usepackage{amsmath} 
\usepackage{amssymb}  
\usepackage{graphicx}
\usepackage{xcolor}
\usepackage{algorithmic}
\usepackage{algorithm}
\usepackage{accents}
\usepackage{enumerate}
\usepackage{nicefrac}
\usepackage{afterpage}
\usepackage{url}
\newcommand{\ubar}[1]
{\underaccent{\bar}{#1}}
\newcommand{\R}{\mathcal{R}}
\newcommand{\B}{\mathcal{B}}
\newcommand{\nbd}[2]{\mathcal{N}^\beta_{#1}(#2)}

\newtheorem{theorem}{Theorem}

\newtheorem{lemma}{Lemma}
\newtheorem{conjecture}{Conjecture}
\newtheorem{proposition}{Proposition}
\newtheorem{definition}{Definition}
\newtheorem{example}{Example}
\newtheorem{remark}{Remark}
\newcommand{\A}{\mathcal{A}
}
\newcommand{\hide}[1]{}
\renewcommand{\a}{a}
 \newcommand{\eop}{{\hfill $\blacksquare$}}

\renewcommand{\P}[2]{\mathcal{P}_{#2}(#1)}

\definecolor{mygreen}{rgb}{0.1, 0.7, 0.1}
\definecolor{mybrown}{rgb}{0.74, 0.27, 0.07}

\newif\ifshowrevisioncolors
\showrevisioncolorsfalse   

\newcommand{\revcolor}[2]{%
  \ifshowrevisioncolors
    \textcolor{#1}{#2}%
  \else
    #2%
  \fi
}

\newcommand{\revseven}[1]{\revcolor{blue}{#1}}
\newcommand{\revten}[1]{\revcolor{mygreen}{#1}}
\newcommand{\reveleven}[1]{\revcolor{red}{#1}}
\newcommand{\revother}[1]{\revcolor{orange}{#1}}

\begin{document}

\begin{frontmatter}

\title{Truncated Noisy Best-Response Algorithms:\\ Toward Game Theoretic Learning with Safety Guarantees\thanksref{footnoteinfo}} 

\thanks[footnoteinfo]{This paper was not presented at any IFAC 
meeting; it extends work originally presented in~\cite{Singh2024}.
Research was sponsored by the Air Force Office of Scientific Research under award number FA9550-23-1-0171 and by the Army Research Office under grant number W911NF-25-1-0239. The views and conclusions contained in this document are those of the authors and should not be interpreted as representing the official policies, either expressed or implied, of the AFOSR, Army Research Office, or the U.S. Government. The U.S. Government is authorized to reproduce and distribute reprints for Government purposes notwithstanding any copyright notation herein.
Corresponding author P.~N.~Brown. Tel. +1-719-255-3332.}

\author[BY]{Vartika Singh}\ead{singh.vsvartika@gmail.com},    
\author[COS]{Philip N. Brown}\ead{pbrown2@uccs.edu},               

\address[BY]{Blue Yonder, India}
\address[COS]{University of Colorado Colorado Springs, CO 80918, USA}             

\begin{abstract}

We consider a game theoretic approach to solve multi-agent coordination problems with submodular maximization objectives.
It is known for such problems that the Nash equilibria for the corresponding game are always within 50\% of the optimal, but that the equilibria which achieve this worst-case bound are not stable.
To exploit this instability, we propose a family of algorithms which we call \emph{Truncated Noisy Best-Response (TNBR) Algorithms}. 
These algorithms are flexibly characterized by agents asynchronously and stochastically selecting actions from a neighbourhood of their best response payoffs.
We compute bounds on the recurrent classes of TNBR algorithms' associated Markov chains.
Our bounds fall into two categories: first, ``Performance'' bounds ensure that TNBR algorithms \emph{always} have a high-value recurrent state; second, ``Safety'' bounds ensure that TNBR algorithms \emph{never} have arbitrarily-bad recurrent states.
Furthermore, these two types of bounds are linked by a waterbed-like effect: every game with a poor Safety guarantee necessarily has a favorable Performance guarantee. 

\end{abstract}

\end{frontmatter}
\section{INTRODUCTION}

Multi-agent coordination problems have a wide range of applications such as task assignment, resource sharing, channel access control in wireless networks, coverage optimization, network routing, and others (see \cite{Brown,Ferguson2021,Kordonis,liu,Qu}). 
The goal is to maximize a system objective controlled by the actions of all the agents. A natural approach to solve such problems uses game theory (see \cite{marden,martin}) --- a system planner endows the agents with individual utility functions and decision rules that depend only on their local information. The system planner can design the local utility functions in order to drive the equilibria of the game toward a desired system optimal state~(e.g., \cite{Singh2025}). 

It is well-understood that if agent utility functions are selected carefully (i.e., using the \emph{marginal contribution utility design}~\cite{marden}), the resulting game is a \emph{potential game} with the system objective as the potential function (see e.g. \cite{marden_potential}), and it is known that these equilibria can be learned efficiently by distributed algorithms~\cite{Giannakopoulos2024}.
In this case, the system-optimal outcome is also a Nash equilibrium (NE) of the game; however, suboptimal Nash equilibria can exist as well.
An extensive literature has emerged to derive and optimize lower bounds on the quality of these suboptimal equilibria~\cite{Paccagnan2020,PoA}.

When the system objective is submodular maximization and agent payoffs have been designed appropriately, it has been proved that every Nash equilibrium obtains a system objective value within~50\% of optimal~\cite{Vetta2002}.
Furthermore, it is known for these games that a) any NE close to this 50\% bound is not stable, and b) any stable NE obtains significantly more than 50\% of optimal~\cite{Seaton2023a}.
In other words, a stable NE cannot be very bad, and a bad NE cannot be stable. 
This suggests that by adding some noise to agent behaviour, it may be possible to selectively avoid the poor-performing NE without substantially destabilizing the good-performing NE.

In this paper, we extend our previous work in~\cite{Singh2024} and prove positive results in this direction which apply to a large family of distributed algorithms parameterized by a neighbourhood size $\beta\in[0,1]$; we term these \emph{Truncated Noisy Best-Response (TNBR) Algorithms}.
In this family of algorithms, agents compute a best-response ``neighbourhood'' of all actions which obtain a payoff within $\beta$ of their best-response payoff.
The agents asynchronously make randomized selections from these best-response neighbourhoods.
Here, $\beta=0$ resolves to a simple asynchronous best-response Markov chain, for which (possibly suboptimal) NE are absorbing states.
As $\beta$ increases from $0$, one hopes that suboptimal NE become transient while high-performing NE remain recurrent.

In this paper, we study two-player submodular maximization games with appropriate payoffs and system objectives normalized to the interval $[0,1]$.
Our results come in two broad categories: 
\begin{itemize}
    \item \emph{(Performance)} In any such game, every TNBR algorithm has a recurrent action profile with system objective value strictly greater than $\frac{1}{2} + \beta$ (Theorems~\ref{thm_NE},~\ref{thm_abs_state}).
    \item \emph{(Safety)} In any such game and any TNBR algorithm, every recurrent action profile has system objective value at least $\frac{1}{2}-g(\beta)$ where $g(\cdot)$ is a nondecreasing function \revseven{satisfying $g(0)=0$} (Theorems~\ref{thm_abs_state},~\ref{cor_rec_state}).
\end{itemize}

\revseven{
The expressions $\frac{1}{2} + \beta$ and $\frac{1}{2}-g(\beta)$ can be viewed as an extension of classical bounds on worst-case Nash equilibria: when $\beta=0$, both expressions coincide at $\frac12$, which is the known lower bound on the system objective of Nash equilibria in this class of games~\cite{Vetta2002}.
Our work provides a theoretical justification for algorithmic optimism: when agents randomize (locally suboptimally), something \emph{better} than a worst-case equilibrium is visited infinitely often, and arbitrarily-poor states are transient.
}

\revseven{
Note that these types of guarantees on a distributed dynamic update process are rare in the game theory literature.
There are many results on practical equilibrium-seeking dynamics~\cite{Monderer1996}, but these typically guarantee convergence to \emph{some} Nash equilibrium without being able to guarantee convergence to a \emph{good} equilibrium.
If worst-case bounds are known on equilibrium quality (e.g., the Price of Anarchy in submodular maximization games~\cite{Vetta2002}), this represents a safety guarantee without a corresponding performance guarantee.
On the other hand, there are dynamics in games which are known to select only the optimal Nash equilibria (e.g., the logit-response dynamics in potential games~\cite{Alos-Ferrer2010}), but these almost universally do so by making every action profile recurrent.
That is, they represent a performance guarantee without corresponding safety properties.
To our knowledge, our work is the first which provides both.
}

Furthermore, we prove several intriguing subtleties and tradeoffs surrounding the relationship between the Safety and Performance properties, and in many cases an unfavorable Safety bound forces an improvement of the corresponding Performance bound.
For example, Theorem~\ref{thm_abs_state} shows that in a recurrent class which contains {no} system-optimal outcome, if $g(\cdot)$ is very steep (indicating an unfavorable Safety property), then the corresponding Performance bound necessarily improves by the same amount.


%
%
%
We also include a suite of numerical experiments to empirically study the expected value of a specific TNBR algorithm called the Approximate Best Response Algorithm (ABRA).
Here, our experiments indicate that by carefully choosing the specific TNBR algorithm, it is possible to improve upon our baseline bounds substantially.

\begin{figure}
    \centering
    \includegraphics[trim ={3cm 8cm 3cm 9cm}, clip, scale=0.5]{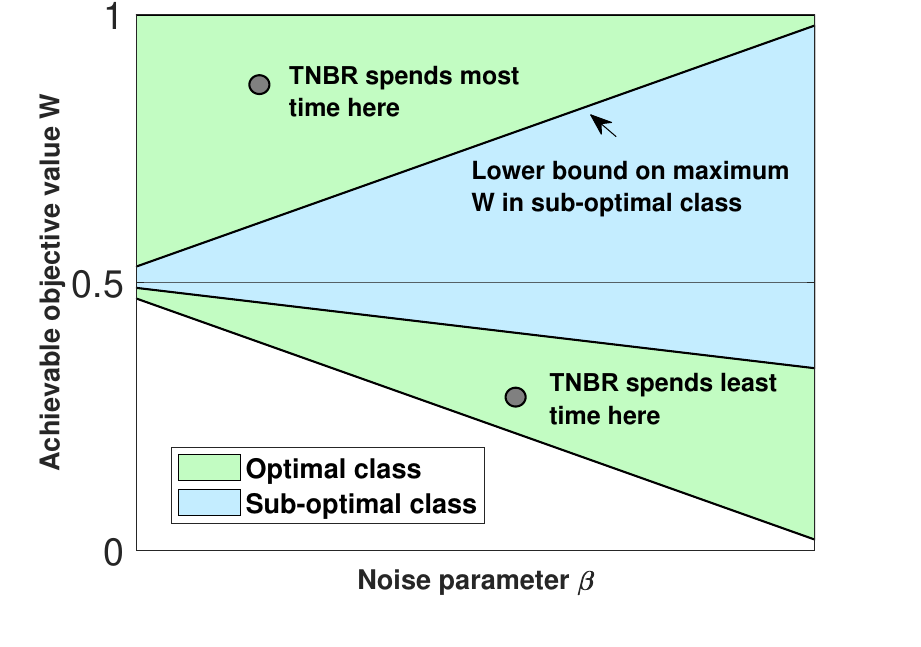}
    \caption{\revten{Conceptual plot of the system objective values of the recurrent states of a TNBR algorithm.
    (i) If an optimal action profile is recurrent (performance), then every action profile below the green region is transient (safety).
    (ii) Otherwise, some action profile above the blue region is recurrent (performance) and every action profile below the blue region is transient (safety).}}
    \label{fig:tradeoff}
\end{figure}

\section{Model}

We consider a game theoretic approach for sub-modular maximization problems. 
A system designer provides two agents with local utility functions and decision rules that depend only on the information available to the agents. 
Given finite sets of elements $E_1$ and $E_2$, agent $i$ has action set $\A_i\subseteq 2^{E_i}$ so that the joint action space is defined by $\A:=\A_1\times\A_2.$
We assume that for each agent $i$, non-participation is feasible; i.e., $\emptyset\in\A_i$.
The optimization problem is specified by a function $W:\A\to[0,1]$.
We assume the function $W$  to be sub-modular, non-decreasing and normalised; that is, for any sets $a, a'$ such that $a \subset a'$ and any element $x\in\A$,  we have
\begin{eqnarray}\label{eqn_sub_mod}
    W(a\cup\{x\}) - W(a) &\ge &  W(a'\cup\{x\}) - W(a'), \nonumber\\
    W(a')\ge W(a) &\mbox{and}& W(\varnothing)=0.
\end{eqnarray}  
Frequently, we slightly abuse notation and write $W(a,b)$ to denote $W((a,b))$ where $a\in\A_1$ and $b\in\A_2$.

Any agent $i$ can observe the action of the other agent but is unaware of their action set $\A_{-i}$, thus cannot directly compute the optimal action profile.  
The system designer aims to design the utility functions and decision rules such that the agent behaviour is driven towards a desirable outcome.  
In this paper, we select the marginal contribution local utility functions (see \cite{marden_potential}):
\begin{equation}\label{eqn_util}
    U_i(a_i,a_{-i}) := W(a_i,a_{-i}) - W(\emptyset,a_{-i}),
\end{equation}
where $a_i$ represents the action of agent $i$ and $a_{-i}$ represents the  action  of  the other agent. It is clear from \eqref{eqn_util} that the local utility of any agent equals the gain in the system utility when that agent chooses to participate. One may anticipate that when agents act as per their self-interest, and try to improve their local utilities, they may also improve the system objective. 

 The marginal contribution utility structure \eqref{eqn_util} induces a potential game among the agents with the system objective function $W$ as the potential function~\cite{marden_potential}. 
 Accordingly, the maximizer of $W$ is always a pure Nash equilibrium (NE) of this game.
 However, the game may have suboptimal NE as well.
 It is shown in~\cite{Vetta2002} that any NE is guaranteed to be within $\frac{1}{2}$ of the optimal for such games.
 Further,~\cite{Seaton2023a} shows that the `bad' NE (the ones with inferior system objective) are not stable: at an unstable NE, none of the agents lose much by switching to the optimal action, i.e., the optimal action is an approximate best response for any agent. Inspired by this, our paper proposes a family of algorithms that leverages the fragility of the `bad' NE in order to escape them by adding a noise parameter $\beta$. We describe the algorithm in the next section.

\reveleven{Before presenting the algorithmic framework, Lemma~\ref{lem_min_entry} states a key structural property of games in our model: at a high-payoff action profile, if one agent can deviate to \emph{any} low-payoff action profile, the other agent necessarily has \emph{only} high-payoff deviations. 
We will use this property repeatedly to prove the paper's main results.}

 \reveleven{\begin{lemma}\label{lem_min_entry}
    If $W$ is sub-modular, non-decreasing and normalised, i.e. satisfies \eqref{eqn_sub_mod}, then for any $\a=(a_1,a_2)$,
    \vspace{-2mm}
    \begin{equation*}
        W(\a)  \le  W(a_1,\tilde{a}_2) +W(\tilde{a}_1,a_2) \mbox{ for all } \tilde{a}_1 \in \A_1,\ \tilde{a}_2 \in \A_2.
    \end{equation*}
\end{lemma}
\noindent\textbf{Proof:} From \eqref{eqn_sub_mod},  we have
\begin{eqnarray*}
    W(a_1,a_2) - W(a_2) &\le&W(\varnothing,a_{1}) - W(\varnothing) = W(a_1), \\
    W(\a)  \le W(a_1)+W(a_2) &\le & W(a_1,\tilde{a}_2) +W(\tilde{a}_1,a_2).
\end{eqnarray*} In the above, the first inequality follows since $W$ is sub-modular. The second inequality follows since $W(\varnothing)=0$, and the last inequality follows from the non-decreasing nature of $W$, \revseven{completing the proof.} \eop}

\hide{\begin{lemma}\label{lem_min_entry}
    If $W$ is sub-modular, non-decreasing and normalised, i.e. satisfies \eqref{eqn_sub_mod}, then for any $\a=(a_1,\dots,a_n)$,
    $$W(\a)  \le  W({a}_i,\tilde{a}_{-i}) +W(\tilde{a}_i,a_{-i}), $$
    for any $i$, $\tilde{a}_i \in \A_i$ and $\tilde{a}_{-i}$.
\end{lemma}
\noindent\textbf{Proof:} Fix $i$. From \eqref{eqn_sub_mod},  we have for any $j$,
\begin{eqnarray*}
    W(a_i,a_j,a_{-\{i,j\}}) - W(a_j,a_{-\{i,j\}}) &\le& W(\varnothing,a_{i}) - W(\varnothing), \\
    W(\a) - W(a_j,a_{-\{i,j\}}) & \le &W(a_i),\\
    W(\a) - W(\tilde{a}_i,a_{-i}) & \le &W(a_i,\tilde{a}_{-i}).
\end{eqnarray*} In the above, the first inequality follows since $W$ is sub-modular. The second inequality follows since $W(\varnothing)=0$, and last inequality follows from the non-decreasing nature of $W$. Hence the proof. \eop
}

\section{Truncated Noisy Best Response Algorithms} \label{sec:algo}

From~\cite{Seaton2023a}, since suboptimal NE have low stability, agents may be able to escape them by adding some noise while still considering the best responses of the agents.
We introduce a noise parameter~$\beta \in [0,0.5)$ and allow the agents to choose actions probabilistically from the $\beta$-neighbourhood of their best response.%
\footnote{A noise parameter $\beta\ge 0.5$ does not add any value to the system and makes all  NE recurrent.}
Basically, the agent makes a randomized action selection only from the set of actions which provide utility within $\beta$ of the current best-response payoff.
A high noise parameter $\beta$ may degrade the performance of the system by allowing inferior actions in the $\beta$-neighbourhood, while a low noise parameter may not be sufficient to escape the bad NE.  Our approach allows us to explicitly characterize this trade-off.

Let $\B_i(a_{-i})$ represent the best response set of agent $i$ when the action of the other agent is $a_{-i}$; i.e., $\B_i(a_{-i}):= \arg\max\{U_i(a_i,a_{-i}): a_i \in \A_i\}$. 
Further, define $\nbd{i}{a_{-i}}$ to be the set of actions in the  $\beta$-neighbourhood of the best response set for agent $i$ (with $a^* \in \B_i(a_{-i})$):
\begin{equation*}
  \nbd{i}{a_{-i}} := \{a_i: |U_i(a^*, a_{-i})-U(a_i,a_{-i})| \le \beta\}.
\end{equation*}%
\revten{A TNBR algorithm }proceeds by allowing agents to update their actions asynchronously according to a pre-specified probability distribution on $\nbd{i}{a_{-i}}$.
To specify this probability distribution, we define an \emph{action selection rule} $F(\cdot)$ as an algorithm which inputs a set of actions and outputs a probability distribution over that set.
When $F(\nbd{i}{a_{-i}})$ assigns positive probability to every action in $\nbd{i}{a_{-i}}$, we say that $F$ is a \emph{fully-supported} action selection rule.
Most of our results are agnostic to the specific choice of action selection rule, but for concreteness we present two feasible examples of fully-supported action selection rules here.

\begin{example} \label{ex:logit}
The \emph{logit selection rule} $F^{\rm log}_\gamma$ with rationality parameter $\gamma>0$ assigns a probability distribution to the actions in $\nbd{i}{a_{-i}}$ according to the softmax function.
That is, the probability that agent $i$ selects action $a_i'\in\nbd{i}{a_{-i}}$ is
\begin{equation}
    \Pr\left[a_i=a_i'\right] = \frac{\exp\left(\gamma U_i(a_i',a_{-i})\right)}{\sum_{a_i^*\in \nbd{i}{a_{-i}}}\exp\left({\gamma U_i(a_i^*,a_{-i})}\right)}.
\end{equation}
This rule is derived from the common \emph{log-linear learning} algorithm~\cite{Alos-Ferrer2010}.%
\footnote{
In the $\beta\to1$ limit, a TNBR algorithm with the logit selection rule resolves exactly to log-linear learning.
}
\end{example}

\begin{example} \label{ex:abra}
The \emph{mistakes rule} $F^{\rm mis}_p$ with rationality parameter $p\in[0,1)$ assigns a probability distribution to the actions in $\nbd{i}{a_{-i}}$ by usually selecting uniformly from the best response set, but occasionally selecting a best-response-neighbourhood action uniformly at random.
That is, agent $i$ selects action $a_i$ according to
\begin{equation*}
    a_i =\left\{ \begin{array}{ll}
       a \in \B_i(a_{-i})  &  \mbox{ with probability } p,\\
         a \in \nbd{i}{a_{-i}}  &  \mbox{ with probability } 1-p.\\ 
    \end{array}\right.
\end{equation*}
In the above, $a$ is chosen from $\B_i(a_{-i})$, or $\nbd{i}{a_{-i}}$ uniformly at random.
This rule is inspired by classical mistakes-based learning dynamics (e.g.,~\cite{Kandori1993,Kandori1995,Young1993}).
Using this action selection rule, one obtains an algorithm substantially identical to the Approximate Best Response Algorithm (ABRA) of the conference version of this paper~\cite{Singh2024}.
\end{example}



With the above definitions in place, we can now precisely state the family of algorithms which we study in this paper.
\begin{definition}
\revten{A \emph{truncated noisy best-response algorithm}} is specified by the tuple $(\beta, F(\cdot))$, where $\beta\geq0$ and $F(\cdot)$ is a fully-supported action selection rule.
The agents are initialized at some joint action profile $\a(0) = (a_1(0),a_2(0))$. 
At any $t \ge 1$, some agent $i$ is selected uniformly at random; this agent observes the joint action profile $\a(t-1)$ and then selects its next action by sampling $a_i(t)\sim F(\nbd{i}{a_{-i}(t-1)})$.
\end{definition}
It can be seen that the sequence of joint action profiles $\{\a(t)\}_t$ induced by any TNBR algorithm $(\beta,F)$ evolves as a Markov chain (see e.g.,~\cite{Levin,Norris}) governed by the noise parameter $\beta$ and action selection rule $F$.
This Markov chain need not be irreducible, and may have multiple recurrent classes depending upon the system objective function $W$.
We can now state the first general result of this paper, that the recurrent classes of these algorithms are governed only by $\beta$, and are independent of the choice of action selection rules:
\begin{theorem} \label{thm:recurrent}
Consider any fully-supported action selection rules $F_1$ and $F_2$ and any 2-player submodular maximization game $W$.
For any $\beta\geq0$, let $\frak{R}_1=\{\R_i\}_{i=1}^{k_1}$ and $\frak{R}_2=\{\R_i\}_{i=1}^{k_2}$ be the sets of recurrent classes of TNBR algorithms $(\beta,F_1)$ and $(\beta,F_2)$, respectively when applied to game $W$.
Then $\frak{R}_1=\frak{R}_2$.
\end{theorem}

\noindent \emph{Proof:} 
%
First note that for any action profile $a$ and either player $i$, because $F_1$ and $F_2$ are fully-supported, it holds that 
\begin{equation} \label{eq:support}
{\rm supp}(F_1(\nbd{i}{a_{-i}}))={\rm supp}( F_2(\nbd{i}{a_{-i}}))=\nbd{i}{a_{-i}}.
\end{equation}
Now, consider action profiles $a=(a_1,a_2)$ and $a'=(a_1',a_2')$.
First, suppose that $a'$ is reachable from $a$ under TNBR algorithm $(\beta,F_1)$; that is, there is a sequence of positive-probability transitions leading from $a$ to $a'$ given the action selection probabilities in $(\beta,F_1)$.
Due to~\eqref{eq:support}, this means that $a'$ is reachable from $a$ under TNBR algorithm $(\beta,F_2)$ as well.
Conversely, the same argument gives that if $a'$ is \emph{not} reachable from $a$ under $(\beta,F_1)$, that it is not reachable from $a$ under $(\beta,F_2)$ either.

Now, let $\R$ be a recurrent class of $(\beta,F_1)$; that is, for any $a,a'\in\R$, $a'$ is reachable from $a$ and for any $a''\notin\R$, $a''$ is not reachable from $a$.
By the above arguments, these reachability relationships hold under $(\beta,F_2)$ as well, meaning that $\R$ is a recurrent class of $(\beta,F_2)$.
By swapping $F_1$ and $F_2$ in this argument, the proof is completed.
%
\hfill \eop

We refer to any recurrent class that contains an optimizer of $W$ as an \textit{optimal class} (denoted by $\mathcal{R}_\beta^*$) and  any recurrent class that does not contain the optimizer as a \textit{sub-optimal class} (denoted by $\mathcal{R}_\beta$). The Markov chain may converge to an absorbing state, a sub-optimal class or an optimal class depending upon initial distribution. 
We now provide the qualitative analysis of the system objective over these classes generated under any TNBR algorithm $(\beta,F)$.

\section{The Recurrent Classes of TNBR Algorithms}\label{sec:main_results}

Here, we present a series of qualitative characterizations of the recurrent classes of TNBR Algorithms, parameterized by $\beta$.
Due to Theorem~\ref{thm:recurrent}, all of the results of this section hold for any action selection rule $F$.
Without loss of generality, assume that $(a^1_1,a_2^1)$ is the optimal profile with $W(a^1_1,a_2^1)=1$. \hide{and define
\begin{equation}\label{eqn_x_bar}
   \ubar{x}:= \min\{ W(a_1^j,a_2^1),W(a^1_1,a_2^k): j \in \mathcal{A}_1, k\in \mathcal{A}_2 \}, 
\end{equation}
to be the minimum value of system objective function that can be achieved when any agent unilaterally deviates from the optimal.
\textit{For ease of explanation, say $\ubar{x}$ is achieved by  unilateral deviation by agent 1, i.e., $\ubar{x}=W(a_1^k,a_2^1)$ for some $a_1^k \in \mathcal{A}_1$}. The analysis follows in the exact similar manner if $\ubar{x}$ was achieved by the deviation of agent 2. }

\subsection{Absorbing states improve upon worst-case}
If the sequence of  action profiles $\{\a(t)\}$ generated using a TNBR algorithm $(\beta,F)$ converges to an absorbing state, then we have the following result  with proof in Appendix:

\begin{theorem}\label{thm_NE}
     If $(a_1^*,a_2^*)$ is an absorbing state for Markov chain $\{(a_1(t),a_2(t))\}_t$ generated using TNBR with noise parameter $\beta$, then, i) $(a_1^*,a_2^*)$ is an NE, and  ii) the system objective satisfies $W(a_1^*,a_2^*) > \frac{1}{2}+\beta$. 
\end{theorem} 

The above result implies that the sequence of action profiles $\{(a_1(t),a_2(t))\}_t$ generated using $(\beta,F)$ converges only to the NE that have system objective value greater than $\frac{1}{2}+\beta$.
This provides a significant improvement over the known bound~\cite{Vetta2002} that every NE has system objective at least $\frac{1}{2}$.
Also, note the strict inequality in the theorem statement; this indicates that no TNBR algorithm can ever converge to a worst-case Nash equilibrium, even when the noise parameter $\beta=0$. 
This reaffirms the findings of \cite{Seaton2023a} regarding the fragility of worst-case NEs.

\subsection{Sub-optimal classes contain high-quality NE}
Now, we consider the sub-optimal classes: those recurrent classes which contain no optimal action profile.
We have the following Performance guarantee with proof in Appendix.

{\begin{theorem}\label{thm_abs_state}
If $\mathcal{R}_\beta$ is a sub-optimal class under  TNBR algorithm $(\beta,F)$  then
\begin{equation}\label{eqn_bdd_min_rec}
    \min_{a\in \R_\beta} W(a) >\frac{1}{2}-\frac{\beta}{2}.
\end{equation}
Further, it holds that
  \begin{equation}\label{eqn_new_max_two_p}
      \max_{a  \in \R_\beta} W(a) > \frac{1}{2}+\delta + \beta
  \end{equation}
where $\delta\in[0,\beta/2)$ satisfies $\min_{a \in \R_\beta} W(a) = \frac{1}{2}-\delta$.
\end{theorem}}

From the above theorem, it is guaranteed that even a sub-optimal class always contains a Nash equilibrium with system objective value more than $\frac{1}{2}+\beta$.
Moreover, when the sub-optimal class contains an action profile with system objective less than $\frac{1}{2}$, then a) the quality of the best Nash equilibrium in the class is further improved by the same amount, and b) the system objective of any action profile in any sub-optimal class cannot be worse than $\frac{1}{2}$ by more than $\frac{\beta}{2}$.

While one would typically hope \emph{not} to converge to a sub-optimal class, Theorem~\ref{thm_abs_state} provides a strong mitigating promise: sub-optimal classes always contain a reasonably ``good'' NE, and cannot contain any arbitrarily bad action profiles.

\subsection{Optimal Classes Have a Safety Guarantee}
We now consider any optimal class $\mathcal{R}^*_\beta$ --- the recurrent class of action profiles that contains the optimal action profile $(a^1_1,a^1_2)$.
We derive a Safety guarantee: a lower bound on the system objective value over the optimal class. 
\revten{Towards proving this result (and later Theorem~\ref{thm: waterbed}), we present the following Lemma which specializes Lemma~\ref{lem_min_entry} to the optimal recurrent class.
The proof of Lemma~\ref{thm_rec_state} appears in the Appendix.}
\revten{\begin{lemma}\label{thm_rec_state}
Define $\ubar{x}$ as the following:
\begin{equation}\label{eqn_x_bar}
   \ubar{x}:= \min\{ W(a_1',a_2^1),W(a^1_1,a_2'): a_1' \in \mathcal{A}_1, a_2'\in \mathcal{A}_2 \}, 
\end{equation}
If $\mathcal{R}_\beta^*$ is an optimal class of action profiles under  TNBR algorithm with $(\beta,F)$, then the system objective function satisfies the following:
\begin{equation}
    W(a_1,a_2) \ge \max\{\max\{1-\ubar{x}, \ubar{x}\}-2\beta, \ubar{x} - \beta\}
\end{equation} for all $(a_1,a_2) \in \mathcal{R}^*_\beta$, with $\ubar{x}$ as in \eqref{eqn_x_bar}.  
\end{lemma}
%
The above result provides the worst possible value that the system objective function can take once the sequence $\{\a(t)\}_t$ enters the optimal class. The lemma shows that this lower bound depends upon the particular game in consideration through $\ubar{x}$.}

\revten{We can now state the main `safety' result:}

\begin{theorem}\label{cor_rec_state}
If $\mathcal{R}_\beta^*$ is an optimal class under TNBR algorithm $(\beta,F)$, then the system objective function satisfies the following:
\begin{equation}\label{eqn_wrst_case_optimal}
\min_{a\in \R_\beta^*} W(a) \ge  \frac{1}{2}-\frac{3\beta}{2}. 
\end{equation}
\end{theorem}
%
\revten{\noindent\textbf{Proof:} From Lemma \ref{thm_rec_state}, the minimum value that $W$ can take over $\mathcal{R}_\beta^*$ is $\max\{\max\{1-\ubar{x}, \ubar{x}\}-2\beta, \ubar{x} - \beta\}$. Hence, for $\ubar{x} \ge \frac{1}{2}$,  $W \ge \ubar{x}-\beta \ge \frac{1}{2}-\beta$. For $\ubar{x}<\frac{1}{2}$, $W \ge \max\{1-\ubar{x} - 2\beta, \ubar{x}-\beta\}$. The right hand side is minimized when $\ubar{x}= \frac{1}{2}-\frac{\beta}{2}$, and provides a minimum value of $\frac{1}{2}-\frac{3}{2}\beta$, obtaining~\eqref{eqn_wrst_case_optimal}. \eop}

The above theorem shows that the minimum value that the system objective can take if $\{\a(t)\}_t$ converges to $\mathcal{R}_\beta^*$ can be as bad as in \eqref{eqn_wrst_case_optimal}. 
This may seem like a negative result but we discuss in Section \ref{sec:num_rationality_parameter} that an appropriate action selection rule can limit the time that a TNBR algorithm spends in the action profiles with the worst system objective value. 
Before that, we describe the trade-off between improving the sub-optimal classes and degrading the optimal class in the following.

\subsection{Trade-off between Safety and Performance guarantees}

To this point, we have shown independent Performance (Theorems~\ref{thm_NE},~\ref{thm_abs_state}) and Safety (Theorems~\ref{thm_abs_state},~\ref{cor_rec_state}) guarantees; here we also show that the two are closely linked by a waterbed-like tradeoff effect.
Specifically, if a particular game's optimal classes have a very poor Theorem~\ref{cor_rec_state} safety guarantee~\eqref{eqn_wrst_case_optimal}, then its sub-optimal classes necessarily have a very favorable Theorem~\ref{thm_abs_state} performance guarantee~\eqref{eqn_bdd_min_rec}.
\reveleven{The precise mechanism causing this is somewhat opaque; fundamentally, the effect occurs because manipulating the payoff matrix to force the optimal-class safety guarantee \emph{down} causes certain payoffs elsewhere in the payoff matrix to rise (enforced indirectly by Lemma~\ref{lem_min_entry}), which pushes the suboptimal-class performance guarantee \emph{up.}
}




\begin{theorem} \label{thm: waterbed}
    If $\R_\beta^*$ is an optimal class under  TNBR algorithm $(\beta,F)$ that satisfies
\begin{equation}\label{eqn_lb_opt_tight}
    \min_{a\in \R_\beta^*} W(a) 
    =  \frac{1}{2}- \beta - \delta
\end{equation} for some $\delta \in [0,\frac{\beta}{2}]$, then every sub-optimal class $\mathcal{R}_\beta$ satisfies
\begin{equation}\label{eqn_lb_subopt_opt}
    \max_{a\in \R_\beta} W(a) > \frac{1}{2}+\beta+\delta.
\end{equation} 
\end{theorem}
\textbf{Proof:} From Lemma \ref{thm_rec_state} and equation \eqref{eqn_lb_opt_tight},
\begin{align*}
    \ubar{x}-\beta & \le  \frac{1}{2}- \beta - \delta \implies \ubar{x} \le \frac{1}{2}-\delta.
\end{align*}Using equation \eqref{eqn_low_min_entry}, the above implies $ W(a_1^1,a_2')\ge   \frac{1}{2} +\delta$  for all $ a_2' \in \A_2$. Using similar arguments as \eqref{eqn_2p_13} leads to \eqref{eqn_lb_subopt_opt}. \eop

\revten{
\begin{remark}
In this paper, we provide the TNBR framework and a characterization of several important tradeoffs, and allow ``end-users'' of TNBR to decide how to select $\beta$ and $F(\cdot)$ given their own specific metrics of interest.
Due to the generality of the TNBR framework, a value of $\beta$ which is optimal (with respect to some metric) for one choice of action selection rule may not be optimal for any other choice of action selection rule.
\end{remark}
}

\section{Numerical Experiments} \label{sec: numerical}
As already mentioned, the sequence of action profiles $\{\a(t)\}$ generated using a TNBR algorithm need not converge to a single absorbing state, and may  have multiple recurrent classes depending upon the algorithm's parameters.
Thus, it is natural to  analyse the expected value of the system objective $W$ under achievable stationary distributions. 
For concreteness, in this section we conduct all our experiments with the \emph{mistakes model} action selection rule as defined in Example~\ref{ex:abra}; we call this the Approximate Best Response Algorithm (ABRA).

Let $\Pi$ be the set of stationary distributions of the Markov chain $\{\a(t)\}_t$ concentrated on various recurrent classes and $E_\pi[\cdot]$ represent the expectation under distribution $\pi$. Define the following performance metric:
\begin{equation}\label{eqn_perf_metric}
  \P{W}{\beta,p}  := \frac{\min_{\pi \in \Pi}E_\pi[W]}{\max_{\a} W(\a)} = \min_{\pi \in \Pi}E_\pi[W],
\end{equation}
since $\max_{\a} W(\a)=1$ in our framework. This metric measures the worst case expected system objective as compared to the optimal value.

\subsection{\revten{Experiments on a Small Game}}\label{sec:num_rationality_parameter}
In this section, we numerically show the variation in probability of choosing an optimal action under stationary distribution concentrated over optimal class, $q^*_p$, as the rationality parameter varies. We also plot the expected system objective as a function of rationality parameter. Consider the following system objective function, 
\begin{equation}\label{eqn_W_ABRA_avoid_optimal_class_NE}
    W=\begin{bmatrix}
        1  &   0.8 &   0.6  &  0.5  \\
    0.5 &  0.6  &   0.4  &  0.39 \\
    0.5  &  0.39 &   0.39  &  0.39\\
    0.5 &  0.39  &  0.39   & 0.71
    \end{bmatrix},
\end{equation}which is controlled by actions of two players with action sets, $\A_i=\{a_i^1,a_i^2,a_i^3,a_i^4\}$ for $i =1,2$.
The $(j,k)$-th entry of matrix in \eqref{eqn_W_ABRA_avoid_optimal_class_NE} corresponds to the system objective value achieved under action profile $(a_1^j,a_2^k)$; for example, $W(a^1_1,a^1_2)=1$ is the system optimal. 
Recall that the system objective function is the potential function for the corresponding potential game induced by marginal utilities. It is well known that any NE of the potential game is the local maximizer of the potential function (see \cite{shapley}). Thus, the game induced by  $W$ in \eqref{eqn_W_ABRA_avoid_bad_NE} has two NEs, $(a^1_1,a^1_2)$ with objective value $1$ and $(a_1^4,a_2^4)$ with objective value $0.71$; the noise parameter is fixed to $\beta=0.2$.

In game~\eqref{eqn_W_ABRA_avoid_optimal_class_NE}, as the sequence of action profiles $\{\a(t)\}_t$ generated using ABRA  for $(\beta,p)$ progresses, it may converge to the absorbing state $(a_1^4,a_2^4)$ yielding a system objective value $0.71$ which is strictly more than $\frac{1}{2}+\beta$ as shown in Theorem \ref{thm_NE}. 
Otherwise, $\{\a(t)\}_t$ may  converge to the optimal class, which for $\beta=0.2$ contains the action profiles with the values $\{1,0.8,0.6,0.5,0.4\}$.
The minimum value that the system objective can take is 0.4, which is less than  $\frac{1}{2}$, and is worse than the well known PoA bound in \cite{Vetta2002}.
\begin{figure}
    \centering
    \includegraphics[trim ={ 3.5cm 8cm 3cm 9cm}, clip, scale=0.45]{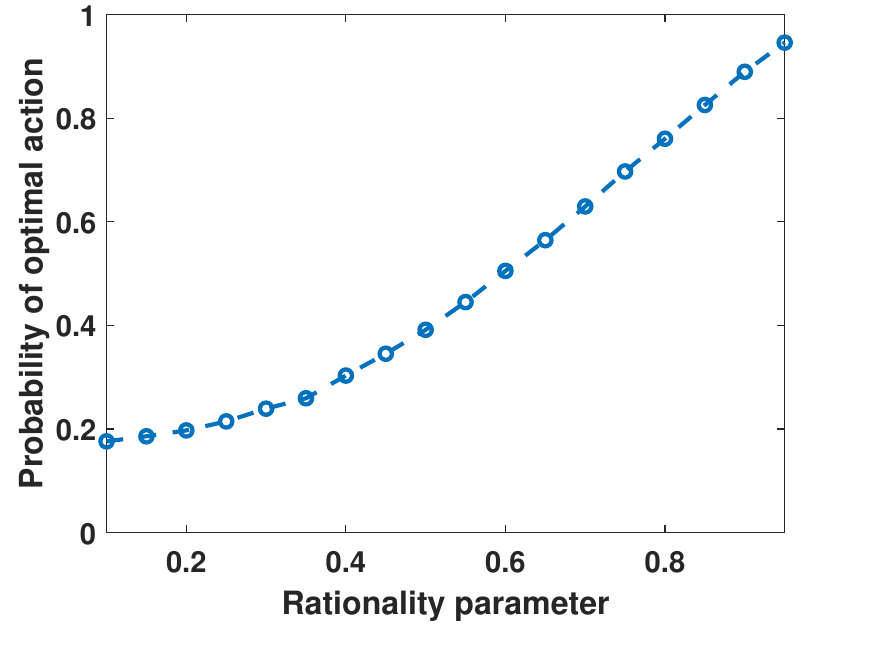}
    \caption{\revten{The probability of choosing an optimal action in an optimal class for $W$ in \eqref{eqn_W_ABRA_avoid_optimal_class_NE} of section \ref{sec:num_rationality_parameter}, when noise parameter $\beta=0.2$.}} 
    \label{fig:stationary_dist}\hide{

    \includegraphics[trim ={ 3.5cm 8cm 3cm 9cm}, clip, scale=0.5]{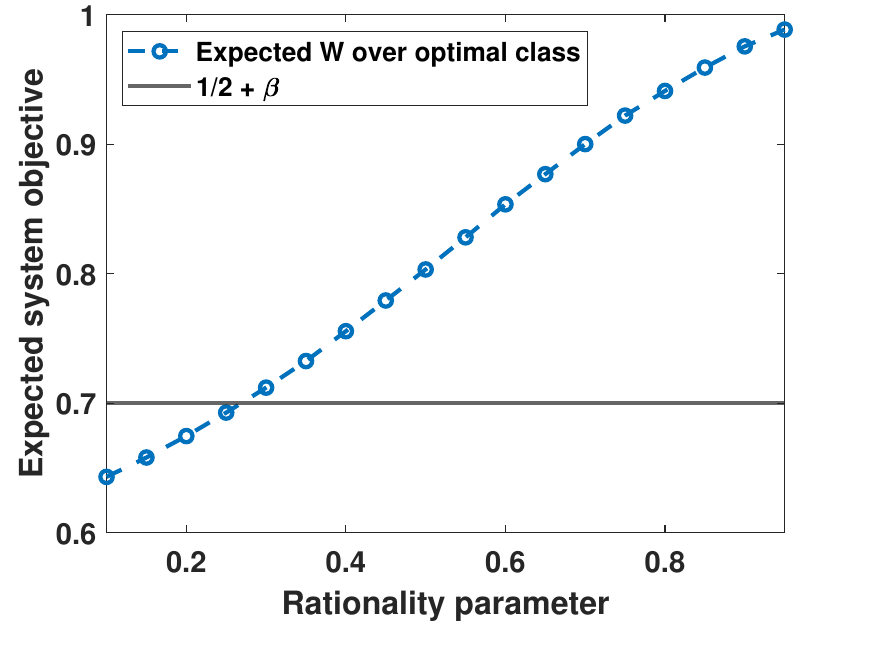}
    \caption{Expected value of system objective under stationary distribution concentrated on optimal class for $W$ in \eqref{eqn_W_ABRA_avoid_optimal_class_NE} of section \ref{sec:num_rationality_parameter}, when  noise parameter $\beta=0.2$. The expected value $E_{\pi^*_p}[W]$  increases with the rationality parameter $p$.}
    \label{fig:expected_W}}
\end{figure}

\begin{figure}
\centering
     \includegraphics[trim ={ 3.5cm 8cm 3cm 9cm}, clip, scale=0.45]{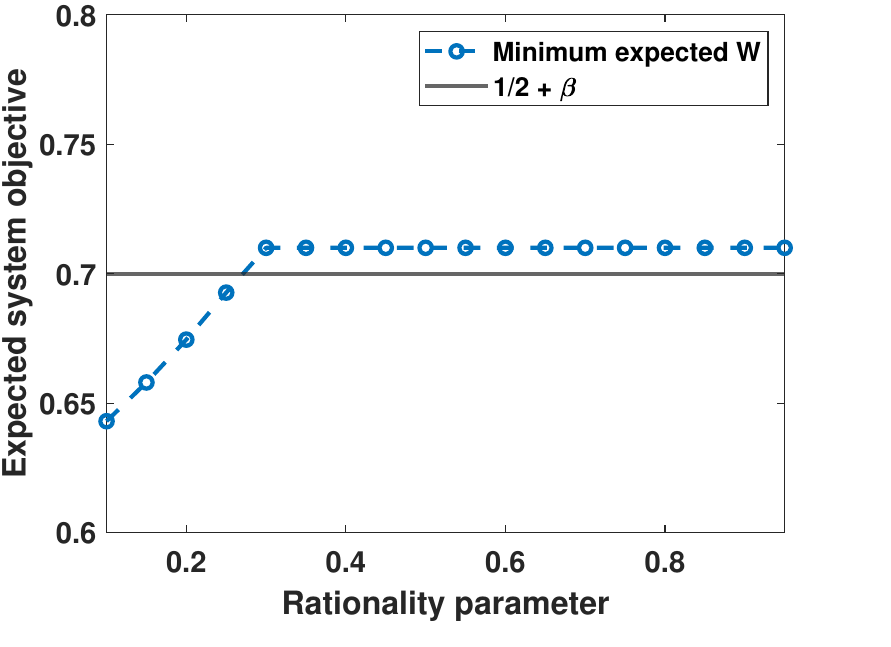}
     \caption{\revten{Minimum expected value of system objective $W$ in \eqref{eqn_W_ABRA_avoid_optimal_class_NE} of section \ref{sec:num_rationality_parameter}, when noise parameter $\beta=0.2$. When rationality parameter $p$ is sufficiently high, this value rises above $\frac{1}{2}+\beta$.}}
     \label{fig:min_exp_W}
 \end{figure}
 
However, in Figure~\ref{fig:stationary_dist}, we observe that the probability of choosing the optimal action  (when $\{\a(t)\}_t$ converges to the optimal class) increases as the rationality parameter $p$ increases.
This implies that ABRA spends most of the time in an optimal action profile if it converges to an optimal class.  
\hide{Figure \ref{fig:expected_W} shows that the expected value of system objective under stationary distribution $\pi^*_p$ concentrated on optimal class, $E_{\pi^*_p}[W]$, also increases with the rationality parameter. Further, $E_{\pi^*_p}[W]> \frac{1}{2}+ \beta$ when rationality parameter $p \ge 0.3$; this is due to  the fact that probability of choosing optimal action, $q^*_p$, increases with $p$. One may anticipate that the minimum expected value of the system objective $ \P{W}{\beta,p} $ defined in see \eqref{eqn_perf_metric} should also improve with the rationality parameter.}
In Figure \ref{fig:min_exp_W}, we plot $ \P{W}{\beta,p} $ as a function of rationality parameter $p<1$ and observe that the minimum expected value of the system objective $ \P{W}{\beta,p} $ in  \eqref{eqn_perf_metric} improves with the rationality parameter $p$ and is above $\frac{1}{2}+\beta$ once $p\ge 0.3$.

 \subsection{Avoiding the bad NE using ABRA}\label{sec:escaping_bad_ne}
Now, we present an example where ABRA avoids a bad NE. We consider the following system objective function,
\begin{equation}\label{eqn_W_ABRA_avoid_bad_NE}
    W=\begin{bmatrix}
        1 & 0.8 & 0.5\\
         0.5   &0.6  &   0.5\\
    0.5 &  0.5  &  0.51
    \end{bmatrix},
\end{equation}which is controlled by actions of two players with action sets, $\A_i=\{a_i^1,a_i^2,a_i^3\}$ for $i =1,2$.  Here $W(a^1_1,a^1_2)=1$ is the system optimal. The game induced by  $W$ in \eqref{eqn_W_ABRA_avoid_bad_NE} has two NEs, the good NE $(a^1_1,a^1_2)$ with objective value $1$ and  the bad NE $(a_1^3,a_2^3)$ with objective value $0.51$.

 In Figure \ref{fig:beta_avoid_bad_NE}, we plot the minimum expected value of system objective $W$, i.e., $\P{\beta,p}{W}$ of \eqref{eqn_perf_metric} 
 as a function of $\beta$ and fix the rationality parameter at $p=0.5$.   When the noise parameter $\beta$ is very small, it is not sufficient to escape the bad NE, and the agents may converge to  the  $(a_1^3,a_2^3)$, which yields minimum possible expected system objective $0.51$. When noise parameter is slightly  increased, e.g., $\beta =0.1$, then the NE $(a^3_1,a^3_2)$ becomes transient, and the optimal action profile becomes the only absorbing state guaranteeing the optimal payoff of 1. However,  when the noise is further increased, the quality of the optimal class degrades, and all the action profiles constitute one optimal class: the addition of sub-optimal action profiles in the recurrent class reduces $\P{\beta,p}{W}$. 

\begin{figure}
    \centering
    \includegraphics[trim ={ 3.5cm 8cm 3cm 9cm}, clip, scale=0.45]{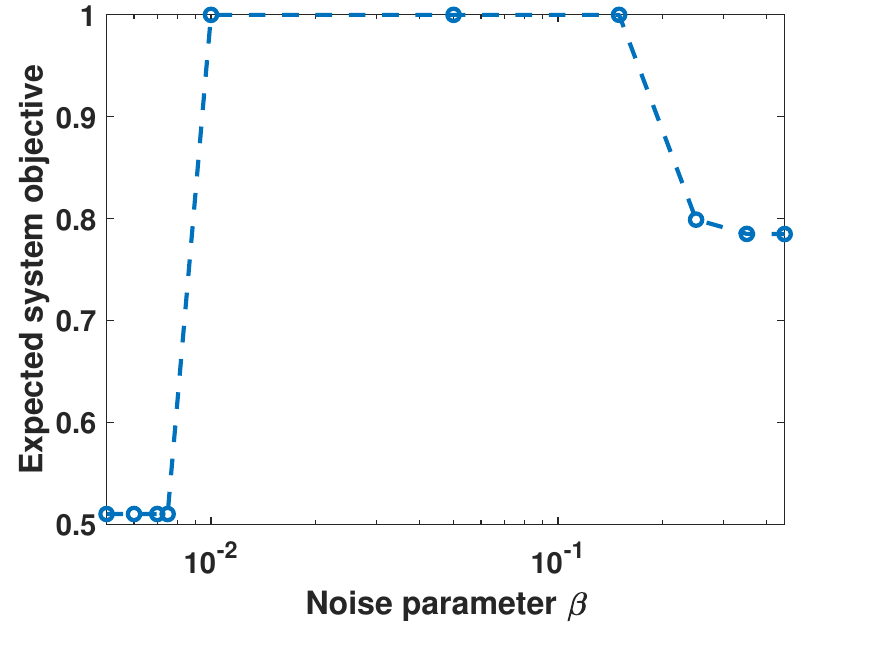}
    \caption{\revten{Minimum expected value of system objective $W$  in \eqref{eqn_W_ABRA_avoid_bad_NE} of section \ref{sec:escaping_bad_ne} as a function of noise parameter $\beta$.
    Adding a small amount of noise escapes the worst NE; however, large amounts of noise reduce system performance.}}
    \label{fig:beta_avoid_bad_NE}
\end{figure}

\subsection{\revten{Experiments on Large Games}}\label{sec:large experiments}

\revten{The experiments in the preceding sections illustrate some of the algorithm's fine-grained behavior in a digestible way.
To complement this, here we present the results of a few experiments on much larger games.
The games in this section are max-weight set cover games (this is a classic submodular problem; see e.g.~\cite{Paccagnan2020}).
A max-weight set cover game has a set of resources $\mathcal{R}$ of positive values $\{v_r\}_{r\in\mathcal{R}}$.
Agent $i$ has action set $\mathcal{A}_i\subset 2^\mathcal{R}$.
The problem's system objective is to maximize the total value of resources which are included in any agent's action set:
\begin{equation}
    W(a_1,a_2) := \sum_{r\in a_1\cup a_2} v_r.
\end{equation}
For these experiments, we generated two instances of max-weight set cover games with $|\mathcal{R}|=9$ and $|\mathcal{A}_i|=50$ for $i\in\{1,2\}$.
Resource values were uniformly sampled (with replacement) from $\{1,\dots,10\}$, and action sets were uniformly sampled (without replacement) from the subsets $2^\mathcal{R}$ containing no more than 4 elements.
To conform to the paper's assumptions, we normalized the system objective to have a maximum of $1$.%
\footnote{For full implementation details, see \url{https://github.com/descon-uccs/tnbr-2026/}.}
}

\revten{Our experiments compare two action selection rules: the logit rule $F_\gamma^{\rm log}$ with $\gamma=10$, and the mistakes rule $F_p^{\rm mis}$ with $p=0$.
That is, our implemented mistakes rule simply performs uniform selection from the best-response neighbourhood $\nbd{i}{a_{-i}}$.
These two rules are convenient, because they each naturally converge to well-known game-theoretic learning algorithms when $\beta$ is at an extreme point.
In particular, $F_\gamma^{\rm log}$ converges to the classic log-linear learning when $\beta\to 1$, and both $F_\gamma^{\rm log}$ and $F_p^{\rm mis}$ converge to asynchronous best response dynamics when $\beta\to0$.
}

\revten{For each generated game and each action selection rule, we sweep $\beta$ from 0 to 1; for each value of $\beta$ and each game instance, we perform 5000 random action initializations, then run each TNBR variant for 20000 agent updates from each initialization.
Then, over the last 5000 agent updates of each run, we calculate
\begin{itemize}
    \item The system objective averaged over the last 5000 agent updates, and
    \item The minimum system objective observed over the last 5000 agent updates.
\end{itemize}
Each of these quantities are recorded as the ``result'' of the particular run.
We then aggregate these last-5000 quantities over all runs, in the following way:
\begin{itemize}
    \item We average all (averaged) system objectives, and
    \item We compute the minimum system objective. This stands as a proxy for the minimum objective value of any recurrent action profile.
\end{itemize}}

\begin{figure} 
    \centering
    \includegraphics[width=.43\textwidth]{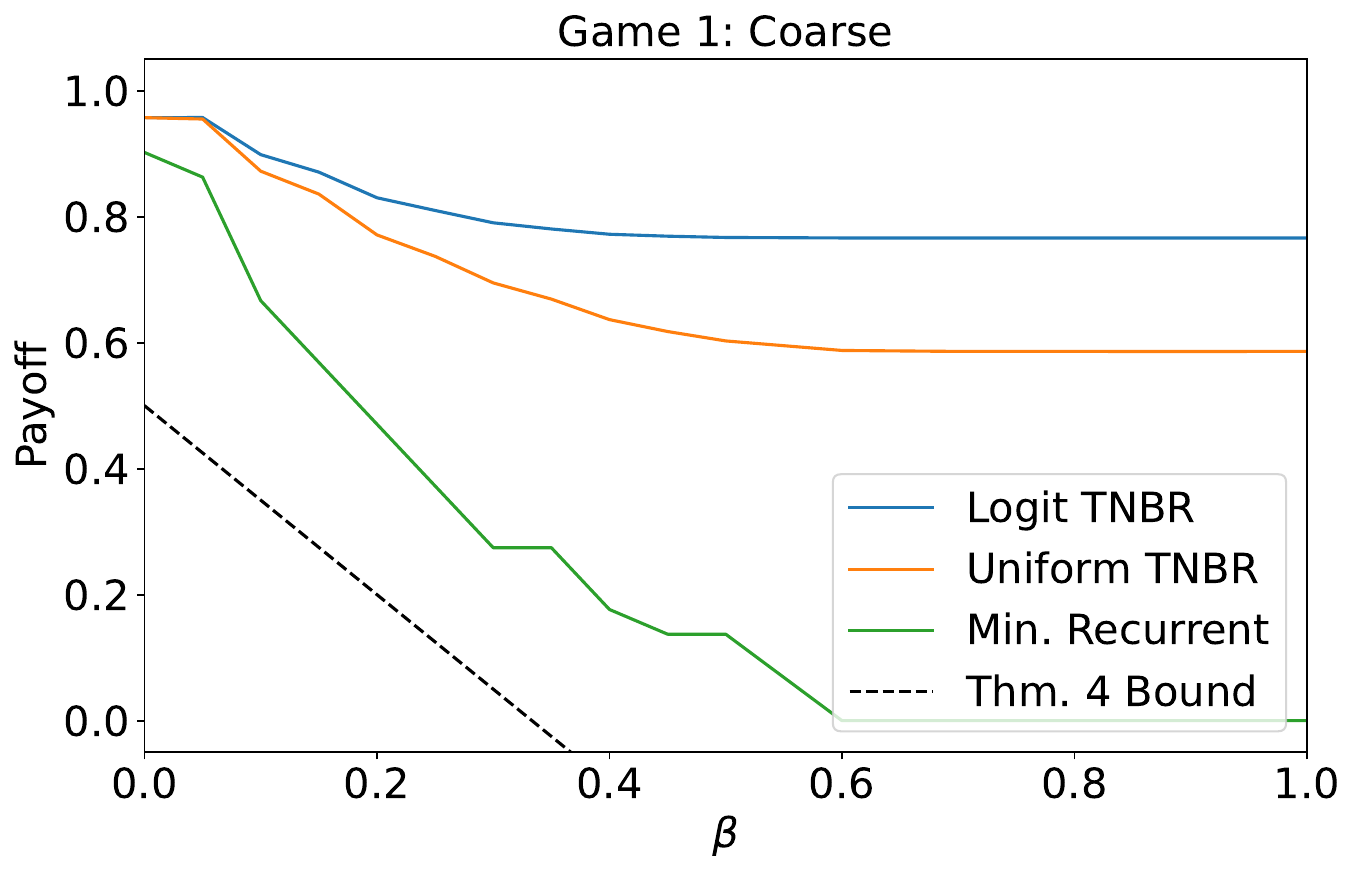}
    \includegraphics[width=.43\textwidth]{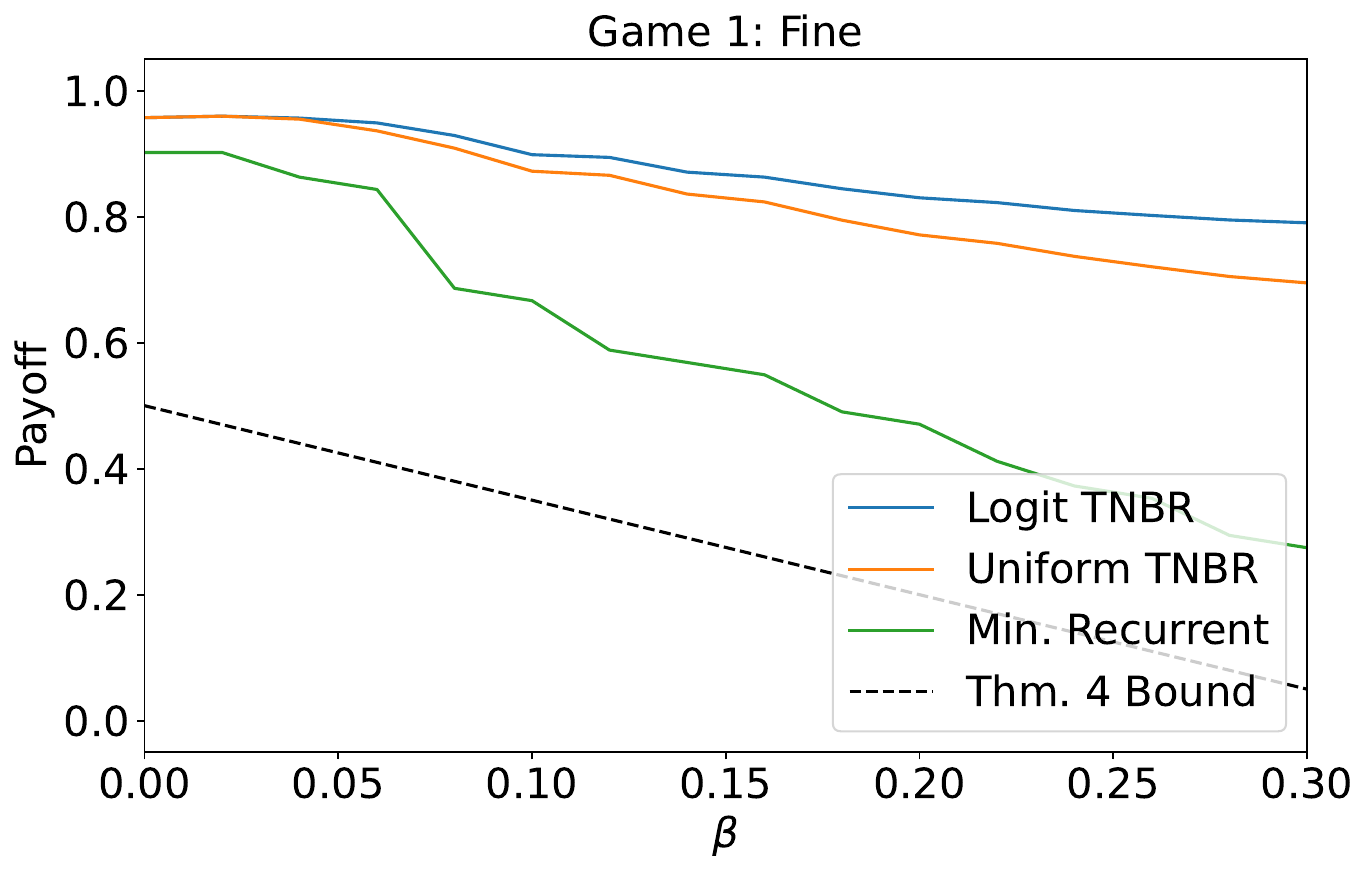}
    \caption{\revten{Large-game experiment results as described in Section~\ref{sec:large experiments} for Game 1, which has worst-case NE of welfare $0.765$ and stability margin $0.641$.
    The TNBR traces are averages over all runs of the average welfare of last 5000 timesteps of each run for the logit and uniform action selection rules, respectively.
    The green trace is the minimum welfare of any recurrent state (estimated empirically).
    Note that the upper plot shows $\beta\in[0,1]$; the lower plot is the same experiment with $\beta\in[0,0.3]$.}}
    \label{fig:game1}
\end{figure}

\begin{figure}
    \centering
    \includegraphics[width=.43\textwidth]{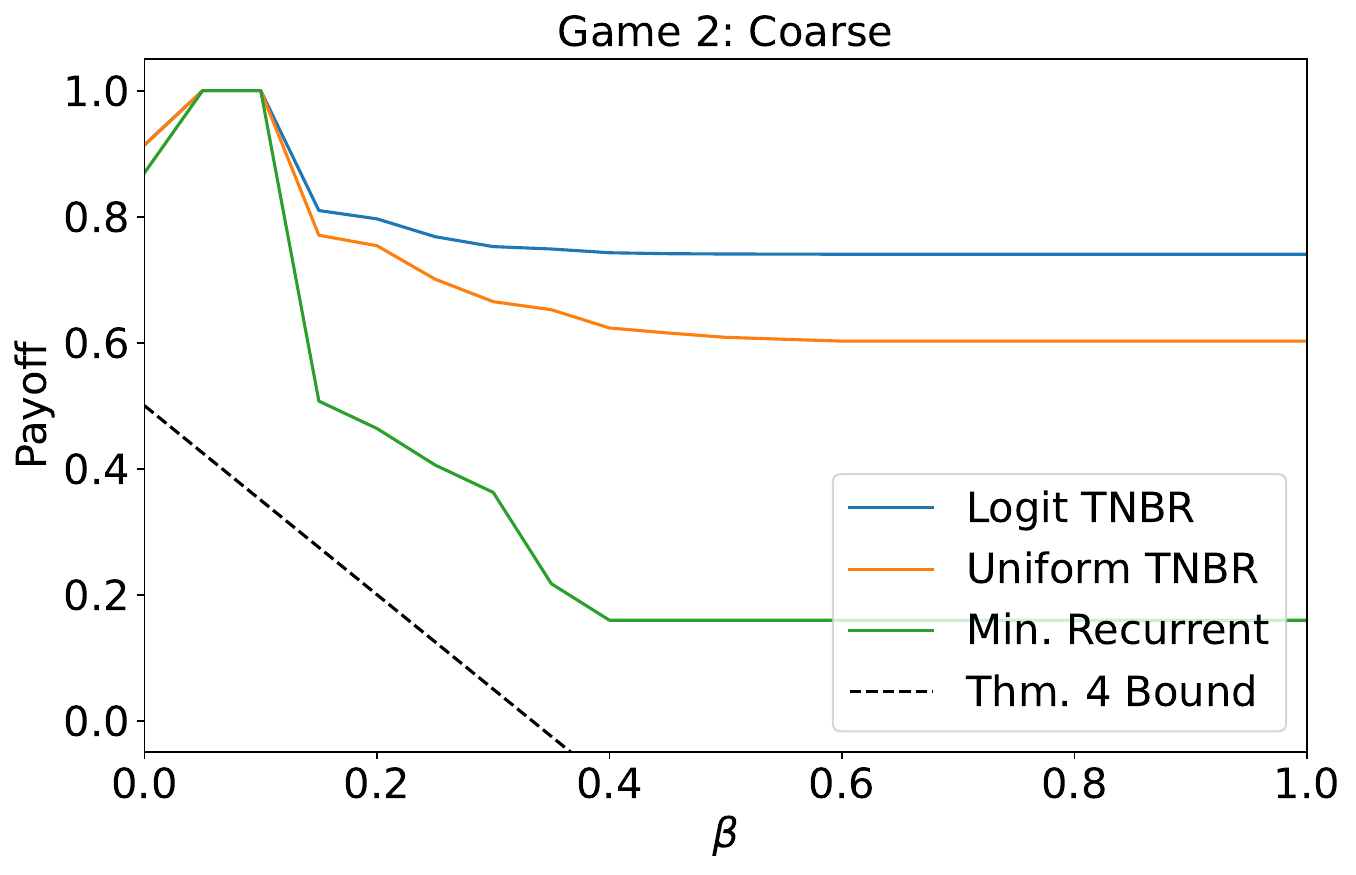}
    \includegraphics[width=.43\textwidth]{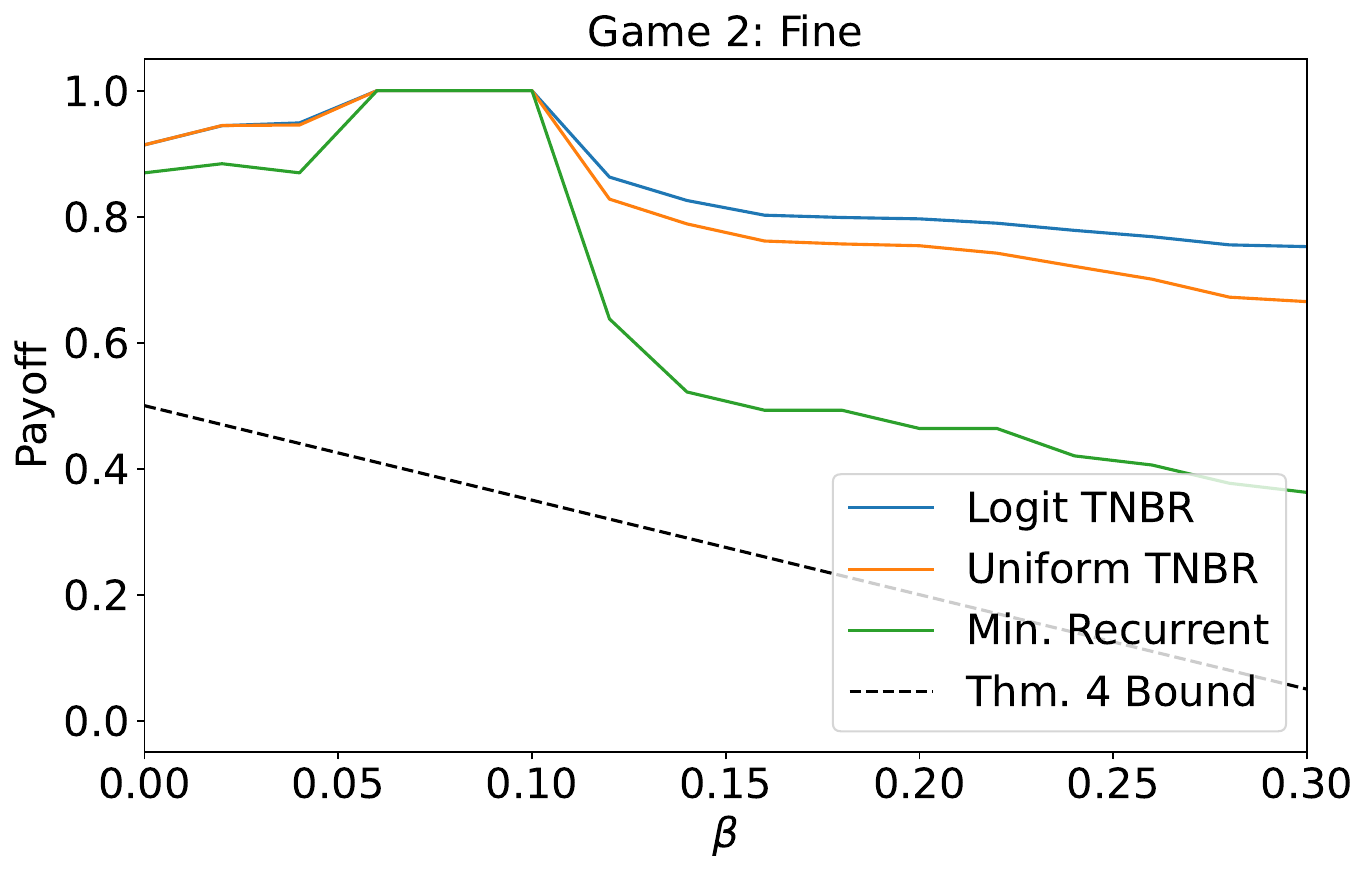}
    \caption{\revten{Large-game experiment results as described in Section~\ref{sec:large experiments} for Game 2, which has worst-case NE of welfare $0.783$ and stability margin $0.037$.
    The TNBR traces are averages over all runs of the average welfare of last 5000 timesteps of each run for the logit and uniform action selection rules, respectively.
    The green trace is the minimum welfare of any recurrent state (estimated empirically).
    Note that the upper plot shows $\beta\in[0,1]$; the lower plot is the same experiment with $\beta\in[0,0.3]$.}}
    \label{fig:game2}
\end{figure}

\revseven{The experimental results are depicted in Figures~\ref{fig:game1} and~\ref{fig:game2} and display a number of interesting features.
Note that we repeat the experiment on two separate games which we identify as ``Game 1'' and ``Game 2.''
These games have similarly low-efficiency worst-case Nash equilibria of welfare $0.765$ and $0.783$, respectively.
They are distinguished by their \emph{stability margin} values, computed according to the definition in~\cite{Seaton2023a}; these values are $0.641$ and $0.037$, respectively.
In particular, the very low stability margin of $0.037$ in Game 2 indicates that when agents are selecting the worst-case Nash equilibrium, either agent can deviate to a system-optimal action without experiencing a large loss in payoff.
This indicates that the worst-case Nash equilibrium lacks stability.
Interestingly, in both example games, the suboptimal and optimal recurrent classes collapse together into a unique recurrent class even for relatively low values of $\beta$.
The main distinction between the two games is the following:
\begin{itemize}
    \item In Game 1 (high stability margin), when suboptimal and optimal recurrent classes collapse together, the resulting unique recurrent class contains enough low-welfare action profiles that the average payoff (over algorithm initializations) immediately begins to decrease with $\beta$.
    \item In Game 2 (low stability margin), the suboptimal and optimal recurrent classes collapse into a unique recurrent class \emph{which contains only an optimal action profile,} at least until $\beta>0.1$.
\end{itemize}
While our experiments cannot tell if this is a statistically significant distinction, it may hint at further structure to be studied and will be the subject of future work.
}

\revten{The experimental results also validate the ``safety guarantee'' of Theorem~\ref{cor_rec_state}; the green trace in each plot of Figures~\ref{fig:game1} and~\ref{fig:game2} (the empirically-estimated minimum welfare of any recurrent state) remains substantially above the black dashed line (the lower bound given in Theorem~\ref{cor_rec_state}). 
}

\section{\revten{Discussion and} Conclusions}
We propose the class \revten{of TNBR algorithms} as applied to 2-agent submodular maximization problems, and we analytically prove worst-case guarantees on the recurrent classes of the resulting Markov chain.
To our knowledge, these guarantees are the first of their kind beyond the simple decades-old bounds provided by Price of Anarchy analysis.
We hope this inspires new questions in the study of performance in multiagent optimization algorithms: future work should investigate more nuanced questions than those of traditional worst-case equilibrium analysis.

\revten{
We hope to extend our result to general $n$-player submodular games.
To date, we have found no straightforward extensions beyond $2$ players.
It is possible to prove a multi-player version of Lemma~\ref{lem_min_entry}, but we have been unable to find a way to use it to prove serviceable bounds such as those presented in Theorem~\ref{cor_rec_state}.
Many of the 2-player results (e.g.,~\eqref{eqn_x_bar}) are facilitated by the fact that any action profile can be reached by any other in at most 2 unilateral deviations; as the number of agents grows, this ``deviation distance'' grows accordingly and may require different mathematical techniques to analyse.
}

\revten{
Last, we note that TNBR algorithms share a common scalability limitation with classical best-response-based algorithms (e.g.,~\cite{Monderer1996}) and other stochastic algorithms with payoff-dependent probability distributions (e.g., logit-response dynamics~\cite{Alos-Ferrer2010}): agents must evaluate their utility functions for \emph{every} available action at every timestep.
For games with large action sets, this becomes computationally prohibitive.
Accordingly, future research could investigate payoff-based or sampling-based approaches to promote better scaling behavior.
}



\section*{APPENDIX: Proofs}

{\noindent\textbf{Proof of Lemma \ref{thm_rec_state}: } 
First, note that $\ubar{x}$ is the minimum value of the system objective function that can be achieved when any agent unilaterally deviates from the optimal. As in the proof of Lemma~\ref{thm_rec_state}, we assume without loss of generality that $\ubar{x}$ is achieved by unilateral deviation by agent 1, i.e., $\ubar{x}=W(\revother{a_1'},a_2^1)$ for some \revother{$a_1' \in \mathcal{A}_1$.} The analysis follows in the exact similar manner if $\ubar{x}$ were achieved by the deviation of agent 2.}

{Fix $\beta$ and consider the problems with non-empty $\mathcal{R}_\beta^*$. Consider any action profile $(a_1,a_2)$ in $\mathcal{R}^*_\beta$. If $a_1 = a^1_1$, then by \eqref{eqn_low_min_entry}, 
\begin{eqnarray*}
    W(a_1,a_2) &\ge& \max\{\ubar{x},1-\ubar{x}\},\\
    &\ge& \max\{\max\{1-\ubar{x}, \ubar{x}\}-2\beta,\ubar{x}-\beta\}.
\end{eqnarray*}
Now, say $a_1 \ne a^1_1$. Observe that for all $\tilde{a}_2 \in \A_2$, the set of action choices for player 1, $\B_1(\tilde{a}_2)\cup\nbd{1}{\tilde{a}_2 }$ always guarantees a payoff more than $\max\{1-\ubar{x},\ubar{x}\}-\beta$ using \eqref{eqn_low_min_entry}. This implies that $(a_1,a_2) \in \mathcal{R}_\beta^*$ under TNBR only if $$W(a_1,a_2')\ge \max\{1-\ubar{x},\ubar{x}\}-\beta \mbox{ for some } a_2' \in \A_2.$$%
Then, the $\beta$-neighbourhood of best responses for player 2, $\B_2(a_1)\cup \nbd{2}{a_1 }$ guarantees a payoff more than $\max\{\max\{1-\ubar{x},\ubar{x}\}-2\beta,\ubar{x}-\beta\}$.  Since the agents can never move to an action profile worse than the ones in $\beta$-neighbourhood under TNBR algorithms, and since the above is true for any $(a_1,a_2)$, \revseven{the proof is complete}. \eop}

\noindent\textbf{Proof of Theorem \ref{thm_NE}:} Let $\ubar{x}$ be defined as in~\eqref{eqn_x_bar}; that is, the minimum value of system objective function that can be achieved when any agent unilaterally deviates from the optimal. As in the proof of Lemma~\ref{thm_rec_state}, we assume without loss of generality that $\ubar{x}$ is achieved by unilateral deviation by agent 1, i.e., \revother{$\ubar{x}=W(a_1',a_2^1)$ for some $a_1' \in \mathcal{A}_1$.} The analysis follows in the exact similar manner if $\ubar{x}$ was achieved by the deviation of agent 2. 

For any $a_2' \in \A_2$, first observe by the definition of $\ubar{x}$ in \eqref{eqn_x_bar} and Lemma \ref{lem_min_entry} that
\begin{eqnarray}\label{eqn_low_min_entry}
     W(a_1^1,\revother{a_2'}) &&+\ \ubar{x} \ \ge\  W(a_1^1,a_2^1) =1, \mbox{ thus}\nonumber\\
   W(a_1^1,\revother{a_2'}) &\ge&  1 -\ubar{x}, \mbox{ and again by \eqref{eqn_x_bar}},\nonumber\\ 
   W(a_1^1,\revother{a_2'})&\ge& \max\{\ubar{x},1-\ubar{x}\}\ \ge\  \frac{1}{2} \ \forall \ \revother{a_2' \in \A_2,}
\end{eqnarray}because $\max\{\ubar{x},1-\ubar{x}\}\ge  \frac{1}{2}$ for all $ \ubar{x} \in [0,1]$.
If $(a_1^*,a_2^*)$ is an absorbing state for the Markov chain generated using TNBR algorithms with noise parameter $\beta$, then it is trivially an NE, and $
    W(a_1^*,a_2^*) - \beta > W(a^1_1,a_2^*) \ge \frac{1}{2} $ by \eqref{eqn_low_min_entry}.  \eop

\noindent\textbf{Proof of Theorem \ref{thm_abs_state}:} Let $(a_1,a_2) \in \R_\beta$. Then\footnote{If any of \eqref{eqn_2p_12} or \eqref{eqn_2p_13} do not hold, then $(a_1,a_2)$ can lead to the optimal action profile $(a_1^1,a_2^1)$, violating the definition of the sub-optimal class. If~\eqref{eqn_2p_6} and \eqref{eqn_2p_7} both do not hold then $(a_1,a_2)$ cannot be in $\R_\beta$.}
\begin{eqnarray}
    W(a_1^1,a_2) &<& \max_{a \in \A_1} W(a, a_2) -\beta,\label{eqn_2p_12} \\
  \mbox{and} \ \  W(a_1,a_2^1) &<& \max_{b \in \A_2} W(a_1, b) -\beta, \label{eqn_2p_13}
\end{eqnarray}and at least  one of the following must hold:
\begin{eqnarray}
    W(a_1,a_2) &\ge& \max_{a \in \A_1} W(a, a_2) -\beta,\label{eqn_2p_6} \\
  \mbox{and/or} \ \  W(a_1,a_2) &\ge& \max_{b \in \A_2} W(a_1, b) -\beta. \label{eqn_2p_7}
\end{eqnarray}Say \eqref{eqn_2p_6} holds. Then by \eqref{eqn_low_min_entry} and \eqref{eqn_2p_12},
\begin{equation}\label{eqn_15_holds}
  \frac{1}{2}\ \le \ W(a_1^1,a_2) < \max_{a \in \A_1} W(a, a_2) - \beta \ \le \  W(a_1,a_2).
\end{equation}Thus, we only need to consider the case when  \eqref{eqn_2p_7} holds and \eqref{eqn_2p_6} does not -- we must have some $\tilde{b} \in \A_2$ with $(a_1,\tilde{b}) \in \R_\beta$ such that\footnote{When \eqref{eqn_2p_6} does not hold, $(a_1,a_2)$ is not reachable by action of player~1, thus it much be reachable by action of player~2 -- thus, there must be some action profile $(a_1,\tilde{b})\in \R_\beta$ reachable by action of player 1 such that player 2 can reach $(a_1,a_2)$ from $(a_1,\tilde{b})$.},
\begin{eqnarray}\label{eqn_2p_11}
   W(a_1,\tilde{b}) &\ge &  \max_{a \in \A_1} W(a,\tilde{b}) - \beta \ > \ W(a^1_1,\tilde{b}).
   \end{eqnarray}%
We now prove \eqref{eqn_bdd_min_rec} using contradiction. Say  $W(a_1,a_2)= \frac{1}{2}-\frac{\beta}{2}-\epsilon$ for some $\epsilon\ge 0$. 
Using equations \eqref{eqn_2p_13} and \eqref{eqn_2p_7}, 
\begin{eqnarray}
W(a_1,a^1_2) &<& \max_{b \in \A_2} W(a_1,b) - \beta \ \le \ \frac{1}{2} - \frac{\beta}{2} - \epsilon. \hspace{3mm}\label{eqn_2p_9}
\end{eqnarray}By Lemma \ref{lem_min_entry}, $W(a_1^1,\tilde{b}) \ge 1 - W(a_1,a_2^1) > \frac{1}{2} + \frac{\beta}{2} + \epsilon  $. Further by \eqref{eqn_2p_11} and \eqref{eqn_2p_9}, we have the contradiction,
\begin{eqnarray*}
 \frac{1}{2} + \frac{\beta}{2} + \epsilon  &<& W(a_1^1,\tilde{b}) \ <\  W(a_1,\tilde{b}) \ \le\ \frac{1}{2} + \frac{\beta}{2} -\epsilon.
\end{eqnarray*}

Towards proving \eqref{eqn_new_max_two_p}, say  $W(a_1,a_2) = \frac{1}{2} - \delta$ for $0 \le \delta < \frac{\beta}{2}$. \hide{Let  $(a_1,a_2) \in \R_\beta$ be such that $W(a_1,a_2) = \frac{1}{2} - \delta$ for $0 \le \delta < \frac{\beta}{2}$. From the definition of $\R_\beta$, at least one of the \eqref{eqn_2p_6} or \eqref{eqn_2p_7} must hold. First consider the case when \eqref{eqn_2p_6} holds. This implies, further using  \eqref{eqn_low_min_entry} and \eqref{eqn_2p_12},
\begin{eqnarray*}
  \frac{1}{2}\ \le \ W(a_1^1,a_2) &<& \max_{a \in \A_1} W(a, a_2) - \beta \ \le \ \frac{1}{2}  - \delta,
\end{eqnarray*}which provides a contradiction. Thus, $\min_{a \in \R_\beta} W(a) > \frac{1}{2}$ whenever \eqref{eqn_2p_6} holds. Now consider the only remaining case, when \eqref{eqn_2p_7} holds and \eqref{eqn_2p_6} does not.}By \eqref{eqn_15_holds}, equation \eqref{eqn_2p_6}  must not hold, thus \eqref{eqn_2p_7} should be true. By \eqref{eqn_2p_13}  and  \eqref{eqn_2p_7}, 
\begin{eqnarray}
   \frac{1}{2}-\delta = W(a_1,a_2) \ \ge\  \max_{b \in \A_2} W(a_1, b) -\beta &>&  W(a_1,a^1_2).\nonumber
\end{eqnarray}Thus, by Lemma \ref{lem_min_entry}, $ W(a^1_1,a_2) > \frac{1}{2} +\delta$.
\hide{\begin{eqnarray}\label{eqn_2p_3}
     W(a^1_1,a_2) &>& \frac{1}{2} +\delta.
\end{eqnarray}}Using \eqref{eqn_2p_12}, 
\begin{eqnarray*}
  \frac{1}{2} +\delta \ <\  W(a_1^1,a_2) &<& \max_{a \in \A_1} W(a, a_2) -\beta \\
  \max_{a \in \A_1} W(a, a_2)&>&\frac{1}{2} + \delta + \beta,
\end{eqnarray*}and since $\arg\max_{a \in \A_1} W(a, a_2) \in \R_\beta$ whenever $(a_1,a_2) \in \R_\beta$, \eqref{eqn_new_max_two_p} is true.   \eop

\bibliographystyle{plain}        
\bibliography{vartikaLibrary,library}    

@inproceedings{Giannakopoulos2024,
archivePrefix = {arXiv},
arxivId = {2211.07547},
author = {Giannakopoulos, Yiannis and Grosz, Alexander and Melissourgos, Themistoklis},
booktitle = {51st International Colloquium on Automata, Languages, and Programming (ICALP 2024)},
doi = {10.4230/LIPIcs.ICALP.2024.72},
eprint = {2211.07547},
month = {jul},
title = {{On the Smoothed Complexity of Combinatorial Local Search}},
url = {http://arxiv.org/abs/2211.07547},
year = {2024}
}

@inproceedings{Singh2024,
author = {Singh, Vartika and Brown, Philip N.},
booktitle = {2024 IEEE 63rd Conference on Decision and Control (CDC)},
doi = {10.1109/CDC56724.2024.10886194},
isbn = {979-8-3503-1633-9},
month = {dec},
pages = {1745--1750},
publisher = {IEEE},
title = {{ABRA: An algorithm which cannot converge to low-quality Nash equilibria}},
url = {https://ieeexplore.ieee.org/document/10886194/},
year = {2024}
}

@article{Kandori1993,
author = {Kandori, Michihiro and Mailath, George J and Rob, Rafael},
doi = {10.2307/2951777},
isbn = {00129682},
issn = {0012-9682},
journal = {Econometrica},
number = {1},
pages = {29--56},
pmid = {2951777},
title = {{Learning, Mutation, and Long Run Equilibria in Games}},
volume = {61},
year = {1993}
}

@article{Kandori1995,
  title={Evolution of equilibria in the long run: A general theory and applications},
  author={Kandori, Michihiro and Rob, Rafael},
  journal={Journal of Economic Theory},
  volume={65},
  number={2},
  pages={383--414},
  year={1995},
  publisher={Elsevier}
}

@article{Young1993,
author = {Young, H. Peyton},
doi = {10.1007/sll229-006-9034-z},
isbn = {0792335198},
journal = {Econometrica},
number = {1},
pages = {57--84},
title = {{The Evolution of Conventions}},
volume = {61},
year = {1993}
}

@inproceedings{Vetta2002,
author = {Vetta, Adrian},
booktitle = {The 43rd Annual IEEE Symposium on Foundations of Computer Science, 2002. Proceedings.},
doi = {10.1109/SFCS.2002.1181966},
isbn = {0-7695-1822-2},
issn = {0272-5428},
pages = {416--425},
pmid = {21930273},
title = {{Nash equilibria in competitive societies, with applications to facility location, traffic routing and auctions}},
url = {http://ieeexplore.ieee.org/lpdocs/epic03/wrapper.htm?arnumber=1181966},
year = {2002}
}

@article{Monderer1996,
author = {Monderer, Dov and Shapley, Lloyd S.},
doi = {10.1006/game.1996.0044},
isbn = {0899-8256},
issn = {08998256},
journal = {Games and Economic Behavior},
number = {1},
pages = {124--143},
pmid = {17556540},
title = {{Potential Games}},
url = {http://www.sciencedirect.com/science/article/pii/S0899825696900445},
volume = {14},
year = {1996}
}

@article{Alos-Ferrer2010,
author = {Al{\'{o}}s-Ferrer, Carlos and Netzer, Nick},
doi = {10.1016/j.geb.2009.08.004},
isbn = {0899-8256},
issn = {08998256},
journal = {Games and Economic Behavior},
number = {2},
pages = {413--427},
publisher = {Elsevier Inc.},
title = {{The logit-response dynamics}},
url = {http://dx.doi.org/10.1016/j.geb.2009.08.004},
volume = {68},
year = {2010}
}

@INPROCEEDINGS{Singh2025,

  author={Singh, Vartika and Wesley, Will and Brown, Philip N.},

  booktitle={2025 American Control Conference (ACC)}, 

  title={Optimal Utility Design with Arbitrary Information Networks}, 

  year={2025},

  volume={},

  number={},

  pages={2895-2900},

  doi={10.23919/ACC63710.2025.11107761}}

@article{Paccagnan2020,
author = {Paccagnan, Dario and Chandan, Rahul and Marden, Jason R.},
doi = {10.1109/TAC.2019.2961995},
issn = {0018-9286},
journal = {IEEE Transactions on Automatic Control},
month = {nov},
number = {11},
pages = {4616--4631},
title = {{Utility Design for Distributed Resource Allocation—Part I: Characterizing and Optimizing the Exact Price of Anarchy}},
url = {https://ieeexplore.ieee.org/document/8941319/},
volume = {65},
year = {2020}
}

@article{Ferguson2021,
archivePrefix = {arXiv},
arxivId = {2102.09655},
author = {Ferguson, Bryce L. and Brown, Philip N. and Marden, Jason R.},
doi = {10.1109/TAC.2021.3088412},
eprint = {2102.09655},
issn = {0018-9286},
journal = {IEEE Transactions on Automatic Control},
month = {jun},
number = {6},
pages = {2729--2742},
title = {{The Effectiveness of Subsidies and Tolls in Congestion Games}},
url = {http://arxiv.org/abs/2102.09655 https://ieeexplore.ieee.org/document/9451652/},
volume = {67},
year = {2022}
}

@article{Seaton2023a,
author = {Seaton, Joshua H. and Brown, Philip N.},
doi = {10.1109/LCSYS.2023.3335315},
issn = {2475-1456},
journal = {IEEE Control Systems Letters},
pages = {3573--3578},
title = {{On the Intrinsic Fragility of the Price of Anarchy}},
url = {https://ieeexplore.ieee.org/document/10325647/},
volume = {7},
year = {2023}
}

@article{Qu,
  author    = {Qu, Guannan and Brown, Dave and Li, Na},
  title     = {Distributed greedy algorithm for multi-agent task assignment problem with submodular utility functions},
  journal   = {Automatica},
  volume    = {105},
  pages     = {206--215},
  year      = {2019}
}

@article{liu,
  author    = {Liu, Chunxia and Lu, Kaihong and Chen, Xiaojie and Szolnoki, Attila},
  title     = {Game-theoretical approach for task allocation problems with constraints},
  journal   = {Applied Mathematics and Computation},
  volume    = {458},
  pages     = {128251},
  year      = {2023}
}

@article{Brown,
  author    = {Brown, Philip N. and Seaton, Joshua H. and Marden, Jason R.},
  title     = {Robust networked multiagent optimization: designing agents to repair their own utility functions},
  journal   = {Dynamic Games and Applications},
  volume    = {13},
  number    = {1},
  pages     = {187--207},
  year      = {2023}
}

@article{Kordonis,
  author    = {Kordonis, Ioannis and Dessouky, Maged M. and Ioannou, Petros A.},
  title     = {Mechanisms for cooperative freight routing: Incentivizing individual participation},
  journal   = {IEEE Transactions on Intelligent Transportation Systems},
  volume    = {21},
  number    = {5},
  pages     = {2155--2166},
  year      = {2019}
}

@incollection{marden,
  author    = {Marden, Jason R. and Shamma, Jeff S.},
  title     = {Game theory and distributed control},
  booktitle = {Handbook of Game Theory with Economic Applications},
  volume    = {4},
  pages     = {861--899},
  publisher = {Elsevier},
  year      = {2015}
}

@article{martin,
  author    = {Martin, Javier G. and Muros, Francisco J. and Maestre, José M. and Camacho, Eduardo F.},
  title     = {Multi-robot task allocation clustering based on game theory},
  journal   = {Robotics and Autonomous Systems},
  volume    = {161},
  pages     = {104314},
  year      = {2023}
}

@inproceedings{marden_potential,
  author    = {Leslie, David S. and Marden, Jason R.},
  title     = {Equilibrium selection in potential games with noisy rewards},
  booktitle = {International Conference on NETwork Games, Control and Optimization (NetGCooP 2011)},
  pages     = {1--4},
  year      = {2011},
  organization = {IEEE}
}

@inproceedings{PoA,
  author    = {Papadimitriou, Christos},
  title     = {Algorithms, games, and the internet},
  booktitle = {Proceedings of the Thirty-Third Annual ACM Symposium on Theory of Computing},
  pages     = {749--753},
  year      = {2001}
}

@book{Levin,
  author    = {Levin, David A. and Peres, Yuval},
  title     = {Markov Chains and Mixing Times},
  volume    = {107},
  publisher = {American Mathematical Society},
  year      = {2017}
}

@book{Norris,
  author    = {Norris, J. R.},
  title     = {Markov Chains},
  number    = {2},
  publisher = {Cambridge University Press},
  year      = {1998}
}

@article{shapley,
  author    = {Monderer, Dov and Shapley, Lloyd S.},
  title     = {Potential games},
  journal   = {Games and Economic Behavior},
  volume    = {14},
  number    = {1},
  pages     = {124--143},
  year      = {1996}
}

\end{document}